\documentclass[a4paper,11pt]{amsart}
\usepackage[T1]{fontenc}
\usepackage[english]{babel}
\usepackage{amsfonts, amsmath, amssymb, mathrsfs}
\usepackage{amsthm}
\usepackage{dsfont}
\usepackage{braket}
\usepackage{color}
\definecolor{darkbrown}{rgb}{.5,.2,0}

\usepackage{graphics}
\usepackage[colorlinks=true,urlcolor=darkbrown,citecolor=black,linkcolor=black]
{hyperref}
\def\doi[#1]{\href{https://doi.org/#1}{\texttt{DOI:\,#1}}}
\def\url[#1]{\href{#1}{\texttt{URL:\,#1}}}
\usepackage{enumitem}
\setlist[enumerate]{%
topsep=1ex plus0.5ex minus0.2ex,   
parsep=.5ex plus0.3ex minus0.1ex,  
itemsep=.5ex plus0.5ex minus0.2ex, 
labelindent=1ex,   
labelwidth=*,      
labelsep=1ex,      
leftmargin=*,      
itemindent=2ex,    
rightmargin=0ex,   
listparindent=1ex, 
align=left,        
font=\rmfamily,
font=\upshape}
\theoremstyle{plain}
\newtheorem{Lem}{Lemma}
\newtheorem{Thm}{Theorem}
\newtheorem{Cor}{Corollary}
\newtheorem{Pro}{Proposition}
\theoremstyle{definition}
\newtheorem{Rem}{Remark}
\newtheorem{Exa}{Example}
\newtheorem*{Que*}{Question}
\newcommand{\C}{\mathbb{C}}
\newcommand{\N}{\mathbb{N}}

\newcommand{\R}{\mathbb{R}}
\newcommand{\cD}{\mathcal{D}}
\newcommand{\cH}{\mathcal{H}}
\newcommand{\cK}{\mathcal{K}}
\newcommand{\fs}{\mathfrak{s}}
\newcommand{\fA}{\mathfrak{A}}
\newcommand{\fB}{\mathfrak{B}}

\newcommand{\fS}{\mathfrak{S}}
\newcommand{\fT}{\mathfrak{T}}

\DeclareMathOperator\ext{ext}

\DeclareMathOperator\Tr{Tr}

\begin{document}
\title[Extreme Points of Quantum States with Bounded Energy]
{Extreme Points of the Set of Quantum States with Bounded Energy}
\author{Stephan Weis}
\author{Maksim Shirokov}
\begin{abstract}
We show that for any energy observable every extreme point of the set of quantum
states with bounded energy is a pure state. This allows us to write every state
with bounded energy in terms of a continuous convex combination of pure states of
bounded energy. Furthermore, we prove that any quantum state with finite energy
can be represented as a continuous convex combination of pure states with the
same energy. We discuss examples from quantum information theory.
\end{abstract}
\date{January 20th, 2020}
%
%
%
%
%
%
%
%
\maketitle
%
%
\section{Introduction}
It is a practical, realistic assumption in engineering that signaling states
in communication networks have bounded energy. For example, the light
sent through telecom fibers must be suitable for processing at all parties
in a network. The assumption of an energy constraint is also targeted at
theo\-retical concepts. An unbounded energy would lead to an infinite capacity,
the maximal amount of information a communication channel can transmit
\cite[Chapter~11]{H-SCI}.
\par
Therefore, the optimization over a set of energy-constrained quantum states
is typical in quantum information theory, see for example
\cite[Section 11.6]{H-SCI} or
\cite{B&D,
Giovannetti-etal2014,
HolevoShirokov2006,
Man, 
WildeQi2018,
W-EBN}.
Some of these optimization problems are convex optimization problems and could
be studied analytically provided we understand the convex geometry of the set of
energy-constrained quantum states. Here, we study the extreme points.
\par
In Section~\ref{sec:extreme-points}, we show that every extreme point of the set of
quantum states with bounded energy is a pure state, using the idea that every
extreme point of a hyperplane section of a convex set is a convex combination of
two extreme points of the convex set \cite{Barvinok2002}.
\par
As the set of quantum states with bounded energy is a closed, $\mu$-compact convex
set \cite{HolevoShirokov2006,P&Sh}, we are able to write every state with bounded
energy in terms of a continuous convex combination of pure states with bounded
energy (akin to Choquet's theorem) in Section~\ref{sec:decomposition}. Surprisingly,
although the convex set of states with given energy is not closed, any state in it
can be represented as a continuous convex combination of pure states with the
same energy.
\par
The above results allow us to show that the supremum of any convex function on the
set of states with bounded energy can be taken only over pure states provided that
this function is lower semicontinuous or upper semicontinuous and upper bounded.
This result simplifies essentially definitions of several characteristics used in
quantum information theory and adjacent fields of mathematical physics. These
applications are considered in Section~\ref{sec:applications}.
\par
%
%
\section{The Extreme Points are Pure States}
\label{sec:extreme-points}
Let $\cH$ be a separable Hilbert space. The space $\fT$ of
trace-class operators on $\cH$ is a Banach space with the trace norm
$\|\cdot\|_1$. The real Banach space of self-adjoint trace-class operators
contains the convex cone $\fT^+$ of positive trace-class operators, the convex
set $\fT^1$ of positive trace-class operators with trace at most one, and the
convex set $\fS=\fS(\cH)$ of density operators, comprising the positive
trace-class operators with trace one. We call the density operators synonymously
{\em states}. Note that $\fT^1$ is the pyramid over $\fS$ with apex the zero
operator, and $\fT^+$ is the convex cone over $\fS$.
\par
It is well-known that the cone of positive trace-class operators $\fT^+$ is
closed, and that the subsets $\fS$ and $\fT^1$ are closed and bounded
\cite{Murphy1990}. Recall also that the set of extreme points $\,\ext(\fS)$ of
$\fS$ consists of the projectors of rank one, called {\em pure states}, which
may be written in the form $\ket{\psi}\!\!\bra{\psi}$ where $\psi\in\cH$ is a
unit vector. The set of extreme points $\,\ext(\fT^1)$ of $\fT^1$
comprises the pure states and the zero operator.
\par
We define an energy constraint on the cone of positive trace-class operators
$\fT^+$ using a positive, self-adjoint (possibly unbounded) operator $H$ on
$\cH$. It is well-known that there is a spectral measure $E_H$ on the Borel
$\sigma$-algebra of $[0,\infty)$ such that
\[
H=\int_0^\infty \lambda\, dE_H(\lambda),
\]
see for example Theorem 5.7 of \cite{Schmuedgen2012}. We approximate the positive
operator $H$ by the sequence $HP_n=\int_0^n \lambda\, dE_H(\lambda)$ of
bounded operators, where $P_n=\int_0^n dE_H(\lambda)$ is the spectral
projector of $H$ corresponding to $[0,n]$. We define the functional
\[
f_H: \fT^+\to[0,\infty],
\qquad
A\mapsto\Tr HA=\lim_{n\to\infty}\Tr(HP_nA).
\]
For every $E\in\R$, the set of states with expected energy at most $E$ is
\[
\fS_{H,E} = \{\rho\in \fS : \Tr H\rho\leq E\}.
\]
Let
\[
\fT^1_{H,E} = \{\rho\in \fT^1 : \Tr H\rho\leq E\}.
\]
Writing $f_H(A)=\sup_{n\in\N}\Tr(HP_nA)$, it is easy to show that $f_H$ is
lower semi-continuous. If follows that $\fS_{H,E}$ and $\fT^1_{H,E}$ are
closed convex sets.
\par
The topology of $\fS_{H,E}$ will matter in the next section where we prove the
existence of extreme points. In the present section, we show that every extreme
point is a pure state. With this aim, we now turn to general (not necessarily closed)
convex sets. Let $V$ be a real vector space and let $K\subset V$ be a convex set. A
map $f:K\to\R$ is called an {\em affine map} if
\[
f(\lambda x+\mu y)=\lambda f(x)+\mu f(y),
\qquad
x,y\in K,
\quad
\lambda,\mu\geq 0,
\quad \lambda+\mu=1.
\]
Let $f:K\to\R$ be a non-constant affine map and let $\alpha\in\R$. We call the
convex set $\{x\in K: f(x)=\alpha\}$ a {\em hyperplane section} of $K$.
\par
We slightly modify Barvinok's Lemma III.9.1 in \cite{Barvinok2002}.
The proof remains almost the same.
\par
\begin{Lem}\label{Lemma:affine}
Let $K$ be a convex set that contains no infinite rays and that contains the two
endpoints of every open segment in $K$. Let $f:K\to\R$ be a non-constant affine
map and let $\alpha\in\R$. Then every extreme point of the hyperplane section
$L=\{x\in K: f(x)=\alpha\}$ is a convex combination of at most two extreme points
of $K$.
\end{Lem}
\begin{proof}
Let $x$ be an extreme point of $L$. If $x$ is an extreme point of $K$, the proof is
complete. Otherwise, there are $\widetilde{y}\neq\widetilde{z}$ in $K$ such that
$x=\tfrac{1}2(\widetilde{y}+\widetilde{z})$. Let $y,z\in K$ be the endpoints of the
maximal segment in $K$ that contains the segment $[\widetilde{y},\widetilde{z}]$,
with respect to the partial ordering by inclusion. The maximal segment exists by
the assumptions that $K$ contains no infinite rays and that the endpoints of every
open segment in $K$ lie in $K$.
\par
\begin{figure}
\includegraphics{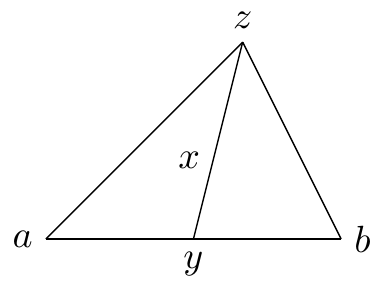}
\caption{Sketch for the proof of Lemma~\ref{Lemma:affine}.}
\label{fig:barvinok}
\end{figure}
We claim that $y$ and $z$ are extreme points of $K$. Assume that $y$ were not an
extreme point of $K$. Then there would be $a\neq b$ in $K$ such that
$y=\tfrac{1}2(a+b)$, see Figure~\ref{fig:barvinok}. As $[y,z]$ is a maximal segment
in $K$, the points $a,b,z$ are
affinely independent. Hence, $x$ is an interior point of the triangle $abz$. This
shows that either $abz\cap L=abz$ or $abz\cap L$ is a segment. In either case, $x$
is an interior point of $abz\cap L$ and this contradicts the assumption that $x$ is
an extreme point of $L$.
\end{proof}
We apply Lemma~\ref{Lemma:affine} to convex sets of positive trace-class operators on
which the functional $f_H$ has finite values. Let
\[
\fT^+_H=\{A\in\fT^+:\Tr HA<\infty\}.
\]
Clearly, $\fT^+_H$ is a convex cone and the restriction
$f_H|_{\fT^+_H}$ is an affine map. Let $\fS_H=\fS\cap\fT^+_H$ and
$\fT^1_H=\fT^1\cap\fT^+_H$.
\par
\begin{Lem}\label{lem:ExtSV}
Every extreme point of $\fS_H$ is a pure state. Every non-zero extreme point
of $\fT^1_H$ is a pure state.
\end{Lem}
\begin{proof}
Let $\rho$ be an extreme point of $\fS_H$ (resp.~of $\fT^1_H$). Let
$\sigma,\tau$ be points in $\fS$ (resp.~in $\fT^1$) and let $\lambda\in(0,1)$
such that $\rho=(1-\lambda)\sigma+\lambda\tau$. We have to prove that
$\sigma=\tau=\rho$. There are three cases. First, if $\Tr H\sigma$ and
$\Tr H\tau$ are finite, then the claim follows from the assumption that $\rho$ is
an extreme point of $\fS_H$ (resp.~of $\fT^1_H$). Secondly, if $\Tr H\sigma=\infty$,
then
\[
\Tr H\rho
=\lim_{n\to\infty}\Tr(HP_n\rho)
=\lim_{n\to\infty}\Tr\big(HP_n((1-\lambda)\sigma+\lambda\tau)\big)
=\infty
\]
contradicts the assumption that $\Tr H\rho<\infty$. The case of $\Tr H\tau=\infty$
is analogous to $\Tr H\sigma=\infty$.
\end{proof}
It will be essential in the proof of Theorem~\ref{thm:ExtSE} that the energy
functional $f_H$ has a finite value at every superposition of two pure states
of finite energy. Therefore, we discuss the problem in detail. We define a
function $\cH\to[0,\infty]$ by
\[
\psi\mapsto\braket{\psi|H|\psi}\doteq
\int_0^{+\infty} \lambda\, d\!\braket{E_H(\lambda)\psi|\psi}.
\]
\par
\begin{Lem}\label{lem:finite-Bloch}
For all $\psi\in\cH$ we have
$\Tr H\ket{\psi}\!\!\bra{\psi}=\braket{\psi|H|\psi}$.
The set of $\psi\in\cH$ for which
$\braket{\psi|H|\psi}<\infty$ is a vector space.
\end{Lem}
\begin{proof}
The monotone convergence theorem and Proposition 4.1 of \cite{Schmuedgen2012}
show
\begin{equation}\label{eq:monotone}
\braket{\psi|H|\psi}
=\lim_{n\to\infty}
\int_0^n \lambda\, d\!\braket{E_H(\lambda)\psi|\psi}
=\lim_{n\to\infty}
\braket{HP_n\psi|\psi},
\qquad\psi\in\cH.
\end{equation}
As $\cH\times\cH\to\C$,
$(\psi,\varphi)\mapsto\braket{HP_n\psi|\varphi}$ is a positive sesquilinear
form, the function $\cH\to\R$, $\psi\mapsto\sqrt{\braket{HP_n\psi|\psi}}$ is a
seminorm for each $n$. The triangle inequality then shows that the set of vectors 
$\psi\in\cH$ for which $\braket{\psi|H|\psi}<\infty$ is a vector space.
Equation \eqref{eq:monotone} shows also that
\begin{align*}
\braket{\psi|H|\psi}
&=\lim_{n\to\infty}
\Tr(HP_n\ket{\psi}\!\!\bra{\psi})
=\Tr H\ket{\psi}\!\!\bra{\psi},
\qquad
\psi\in\cH.
\end{align*}
This completes the proof.
\end{proof}
Note that $\{\psi\in\cH : \braket{\psi|H|\psi}<\infty\}$ is the domain
$\cD(\sqrt{H})$ of the positive square root of $H$, which may be larger 
than the domain of $H$. For all $\psi\in\cD(\sqrt{H})$ we have
$\braket{\psi|H|\psi}=\braket{\sqrt{H}\psi|\sqrt{H}\psi}=\|\sqrt{H}\psi\|^2$,
see for example Proposition 10.5 of \cite{Schmuedgen2012}.
\par
The extreme points of $\fS_{H,E}$ and $\fT^1_{H,E}$ fall into two classes, those
on the hyperplane section
\[
\fA_{H,E}
=\{\rho\in \fT^+_H : \Tr H\rho=E\}
\]
of $\fT^+_H$, and those extreme points $\rho$ with energy $\Tr H\rho<E$.
\par
\begin{Thm}\label{thm:ExtSE}
Every extreme point of $\fS_{H,E}$ is a pure state. Every extreme point
of $\fT^1_{H,E}$ has rank at most one.
\end{Thm}
\begin{proof}
Let $\rho$ be an extreme point of $\fS_{H,E}$ (resp.~of $\fT^1_{H,E}$).
If $\Tr H\rho<E$ we show that $\rho$ is an extreme point of $K=\fS_H$
(resp.~of $K=\fT^1_H$). If $\Tr H\rho=E$, we prove that $\rho$ is a convex
combination of two extreme points of $K$, and we prove that $\rho$ is not
a convex combination of two pure states. Using that, the claim follows
from Lemma~\ref{lem:ExtSV}.
\par
First, let $\Tr H\rho<E$. If $\rho$ is not an extreme point of $K$, then
there is a segment $\fs\subset K$ that contains $\rho$ as an interior point.
As $f_H$ is continuous and affine on $\fs$, there are points $\sigma\neq\tau$
in $\fs$ such that $\rho=\tfrac{1}2(\sigma+\tau)$ and such that
$\Tr H\sigma,\Tr H\tau\leq E$. Hence, the segment $[\sigma,\tau]$ lies in
$\fS_{H,E}$ (resp.~$\fT^1_{H,E}$). This is impossible, as $\rho$ is an extreme
point of $\fS_{H,E}$ (resp.~$\fT^1_{H,E}$), and proves that $\rho$ is an extreme
point of $K$.
\par
Secondly, let $\Tr H\rho=E$. As $\rho$ is an extreme point of $\fS_{H,E}$
(resp.~$\fT^1_{H,E}$), it is {\em a fortiori} an extreme point of
$L=\fS_{H,E}\cap\fA_{H,E}$ (resp.~$L=\fT^1_{H,E}\cap\fA_{H,E}$). Note that $L$
is the hyperplane section $L=\{\rho\in K: f_H(\rho)=E\}$ of $K$. As $K$ is
bounded it contains no infinite rays. Moreover, $f_H|_K$ being an affine map
implies that $f_H$ is bounded on every open segment in $K$. The lower
semi-continuity of $f_H$ then show that the endpoints of every open segment
in $K$ also belong to $K$. Therefore, Lemma~\ref{Lemma:affine} shows that $\rho$
is a convex combination of two extreme points of $K$. It remains to exclude the
case that $\rho$ is a mixture of two pure states in $K$, say
$\ket{\psi}\!\!\bra{\psi}$ and $\ket{\varphi}\!\!\bra{\varphi}$.
In that case, $\rho$ would be an interior point of the three-dimensional Bloch ball
$\fB$ of density operators acting on the span of $\psi$ and $\varphi$. As
$\fB\subset K$ by Lemma~\ref{lem:finite-Bloch}, the intersection $\fB\cap\fA_{H,E}$
is the entire ball $\fB$ or a disk. This contradicts the assumption that $\rho$ is
an extreme point of $L$, as $\fB\cap\fA_{H,E}\subset L$ and as $\rho$ is an interior
point of $\fB\cap\fA_{H,E}$.
\end{proof}
%
%
\section{Pure-State Decomposition Theorem}
\label{sec:decomposition}
If $H$ is an unbounded positive operator on $\cH$ with a discrete spectrum of
finite multiplicity, then the sets $\fS_{H,E}$ and $\fT^1_{H,E}$ are compact.
This has been shown for $\fS_{H,E}$ in \cite{Holevo2004} and can be
shown easily for $\fT^1_{H,E}$ by using Proposition~11 in \cite[Appendix]{AQC}.
\par
If $H$ is an arbitrary positive operator, the sets $\fS_{H,E}$ and
$\fT^1_{H,E}$ are closed but not {\em compact}. Yet, they are {\em $\mu$-compact}
by Proposition 2 in \cite{HolevoShirokov2006} and Proposition 4 in \cite{P&Sh},
respectively. Proposition~5 in \cite{P&Sh} provides generalized assertions of
Krein-Milman's theorem (which we prove directly for $\fS_{H,E}$ below) and of
Choquet's theorem for $\mu$-compact sets. We employ Theorem~\ref{thm:ExtSE} to make
these assertions more explicit.
\par
\begin{Thm}\label{thm:Krein-Milman-and-Choquet}
Let $H$ be an arbitrary positive operator and $E>\inf\sigma(H)$, where $\sigma(H)$
is the spectrum of $H$. Then the set of extreme points $\,\ext\fS_{H,E}$ is nonempty
and closed.
\begin{enumerate}
\item[A (Krein-Milman's theorem).]
The set $\,\fS_{H,E}$ is the closure of the convex hull of $\,\ext\fS_{H,E}$.
\item[B (Choquet's theorem).]
Any state in $\fS_{H,E}$ is the barycenter $\int \sigma\mu(d\sigma)$ of some Borel 
probability measure $\mu$ supported by $\,\ext\fS_{H,E}$.
\end{enumerate}
\end{Thm}
\begin{proof}
The first assertion follows as the infimum of the spectrum of $H$ is equal to the
infimum of $\braket{\varphi|H|\varphi}$ over all unit vectors $\varphi$.
By Theorem~\ref{thm:ExtSE}, the set of extreme points of $\fS_{H,E}$ is the
intersection of $\fS_{H,E}$ with the set of all pure states, both of which are
closed. This shows that $\,\ext\fS_{H,E}$ is closed.
\par
A) By Lemma~\ref{Lemma:density} below, it suffices to show that any finite rank
state $\rho$ in $\fS_{H,E}$ can be represented as a convex combination of pure states
in $\fS_{H,E}$. Let $\cH_{\rho}$ be the support of $\rho$ and $P$ the projector on
this subspace. Since the subspace $\cH_{\rho}$ is finite-dimensional, it is easy to
see that the finiteness of $\Tr H\rho$ implies that $\cH_{\rho}$ belongs to the domain
of $\sqrt{H}$.  It follows that $H_{\rho}=PHP$ is a bounded operator on $\cH_{\rho}$.
The state $\rho$ lies in the set $\fS_{\rho}(E)$ of all states $\sigma$ supported by
$\cH_{\rho}$ and satisfying the inequality $\Tr H_{\rho}\sigma\leq E$. By
Carath{\'e}odory's theorem, $\rho$ is a finite convex combination of extreme points of
$\fS_{\rho}(E)$ which are pure states by Theorem \ref{thm:ExtSE}. It is clear that
$\fS_{\rho}(E)\subseteq\fS_{H,E}$.
\par
B) This assertion follows from Proposition 5 in \cite{P&Sh}, since the set $\fS_{H,E}$
is $\mu$-compact and since $\,\ext\fS_{H,E}$ is closed.
\end{proof}
Note that the closedness of the set $\,\ext\fS_{H,E}$ is not obvious even in the
case of an operator $H$ with discrete spectrum or in the case of $\dim(\cH)<\infty$. 
This property is necessary for the stability \cite{P&Sh} of $\,\fS_{H,E}$.
\par
The set of extreme points $\,\ext\fS_{H,E}$ may be empty if $E=\inf\sigma(H)$ is
the ground-state energy of $H$. The other assertions of
Theorem~\ref{thm:Krein-Milman-and-Choquet} remain true if we replace
$E>\inf\sigma(H)$ with $E=\inf\sigma(H)$.
\par
\begin{Que*}
Under which conditions on the operator $H$ can part B of
Theorem~\ref{thm:Krein-Milman-and-Choquet} be strengthened to the statement that
any state in $\fS_{H,E}$ is a \emph{countable} convex combination of pure states in
$\fS_{H,E}$? This and the arguments of Corollary~\ref{cor:Choquet-same-energy}
below would imply that any state with finite energy is a \emph{countable} convex
combination of pure states with the same energy. In the finite dimensional settings
the existence of such decomposition is shown in \cite{Man}, it also follows directly
from the proof of Theorem~\ref{thm:ExtSE} (second case $\Tr H\rho=E$) and 
Carath{\'e}odory's theorem.
\end{Que*}
We turn to an intriguing representation at constant energy.
\par
\begin{Cor}\label{cor:Choquet-same-energy}
Let $H$ be an arbitrary positive operator. Any state $\rho$ such that
$\Tr H\rho=E<+\infty$ can be represented as follows
\begin{equation}\label{u-rep}
\rho=\int \sigma\mu(d\sigma),
\end{equation}
where $\mu$ is a Borel probability measure supported by pure states such that
$\Tr H\sigma=E$ for $\mu$-almost all $\sigma$.
\end{Cor}
\begin{proof} The assertion B of Theorem~\ref{thm:Krein-Milman-and-Choquet}
implies that equation (\ref{u-rep}) holds for some probability measure $\mu$
supported by the set $\,\ext\fS_{H,E}=\fS_{H,E}\cap\ext\fS$. Since the function
$\sigma\mapsto\Tr H\sigma$ is nonnegative affine and lower semicontinuous, we
have (see, f.i., \cite[the Appendix]{EM})
$$
\int \Tr(H\sigma)\,\mu(d\sigma)=\Tr H\rho=E.
$$
Since $\Tr H\sigma\leq E$ for all $\sigma$ in the support of $\mu$, this equality
implies that $\Tr H\sigma=E$ for $\mu$-almost all $\sigma$.
\end{proof}
Note: Despite the fact that the convex set of states with given energy $E$ is not
closed, any state in it can be represented as a continuous convex combination of
pure states from this set.
\par
The barycenter of a Borel probability measure supported on pure states is known
as a {\em continuous convex combination} or a {\em generalized ensemble} of pure 
states \cite{HolevoShirokov2006}. In this sense, the probability measure $\mu$  
in part B of Theorem~\ref{thm:Krein-Milman-and-Choquet} is a generalized ensemble 
of pure states of bounded energy. In the strict sense, the probability measure $\mu$
in Corollary~\ref{cor:Choquet-same-energy} is not a generalized ensemble of pure 
states of constant energy $E$, as the support of $\mu$ may contain a set of 
$\mu$-measure zero where the energy is smaller than $E$.
\par
\begin{Lem}\label{Lemma:density}
Let $H$ be an arbitrary positive operator and $E\in\R$. The set of
finite-rank states in $\fS_{H,E}$ is dense in $\fS_{H,E}$.
\end{Lem}
\begin{proof} Let $\rho=\sum_{i=1}^{+\infty}p_i\ket{\varphi_i}\!\!\bra{\varphi_i}$
be an infinite-rank state in $\fS_{H,E}$. For any given $n$ let
$$
\rho_n=\sum_{i=1}^{n}p_i\ket{\varphi_i}\!\!\bra{\varphi_i}
 +q_n\ket{\tau_n}\!\!\bra{\tau_n},
\quad
q_n=\sum_{i>n}p_i,
$$
where $\tau_n\in\cH$ is a unit vector such that
$\braket{\tau_n|H|\tau_n}\leq q_n^{-1}\sum_{i>n}p_i\braket{\varphi_i|H|\varphi_i}$.
The existence of the vector $\tau_n$ is clear if $E_0=\inf\sigma(H)$ is an eigenvalue
of $H$, as the infimum of the spectrum of $H$ is equal to the infimum of
$\braket{\varphi|H|\varphi}$ over all unit vectors $\varphi$. This implies also
the existence of $\tau_n$ if $E_0$ is not an eigenvalue of $H$. In that case
\[
\braket{\varphi|H|\varphi}
=\braket{\varphi|H-E_0|\varphi}+E_0
=\|\sqrt{H-E_0}\,\varphi\|^2+E_0
>E_0
\]
holds for all unit vectors $\varphi\in\cH$. It is easy to see that $\rho_n\in\fS_{H,E}$
and that $\rho_n\to\rho$ as $n\to+\infty$.
\end{proof}
%
%
\section{Applications to Quantum Information Theory}
\label{sec:applications}
In this section we consider some applications of our main results in quantum
information theory and mathematical physics. These applications are based on
the following observation.
\par
\begin{Pro}\label{pro:supremum-pure}
Let $H$ be an arbitrary positive operator and $f$ a convex function on the set
$\,\fS_{H,E}$, which is either lower semicontinuous or upper semicontinuous and
upper bounded. Then
\begin{equation}\label{c-2-r}
\sup_{\rho\in\fS_{H,E}}f(\rho)=\sup_{\varphi\in\cH_{E}}f(\ket{\varphi}\!\!\bra{\varphi}),
\end{equation}
where $\cH_{E}=\{\varphi\in\cH\,|\,\braket{\varphi|H|\varphi}\leq E, \|\varphi\|=1\}$.
If the function $f$ is upper semicontinuous and the operator $H$ has unbounded discrete
spectrum of finite multiplicity,
then the supremum on the right-hand side of \eqref{c-2-r} is attained at a unit vector
in $\cH_{E}$.
\end{Pro}
\begin{proof}
By Theorem~\ref{thm:Krein-Milman-and-Choquet}, for any mixed state
$\rho$ in $\fS_{H,E}$ there is a probability measure $\mu$  supported by pure states
in $\fS_{H,E}$ such that
$$
\rho=\int \sigma\mu(d\sigma).
$$
The assumed properties of the function $f$ guarantees (see, f.i.,
\cite[the Appendix]{EM}) the validity of the Jensen inequality
$$
f(\rho)\leq \int f(\sigma)\mu(d\sigma),
$$
which implies the existence of a pure state $\sigma$ in $\fS_{H,E}$ such that
$f(\sigma)\geq f(\rho)$.
\par
If the operator $H$ has unbounded discrete spectrum of finite multiplicity then
the set $\fS_{H,E}$ is compact. Hence, the set of extreme points $\,\ext\fS_{H,E}$
is compact by Theorem~\ref{thm:Krein-Milman-and-Choquet}. This and the above
arguments imply that the first supremum in (\ref{c-2-r}) is attained at a pure
state in $\fS_{H,E}$ (provided that the function $f$ is upper semicontinuous).
\end{proof}
\begin{Rem}\label{rem:energy-equality}
The arguments from the proof of Proposition~\ref{pro:supremum-pure}, when using
Corollary~\ref{cor:Choquet-same-energy} instead of
Theorem~\ref{thm:Krein-Milman-and-Choquet}, show that
\begin{equation*}
\sup_{\rho\in\fS(\cH), \Tr H\rho=E}f(\rho)
=\sup_{\varphi\in\cH_1: \braket{\varphi|H|\varphi}=E}f(\ket{\varphi}\!\!\bra{\varphi})
\end{equation*}
for any convex function $f$ on the set $\,\fS_{H,E}$ which is either lower
semicontinuous or upper semicontinuous and upper bounded. In the above formula
$\cH_{1}$ is the unit sphere in $\cH$.
\end{Rem}
Of course, we may replace the convex function $f$ in
Proposition~\ref{rem:energy-equality} by the concave function $-f$
(and supremum by infimum). This idea is motivated by potential applications, since
many important characteristics of a state in quantum information theory are concave
lower semicontinuous and nonnegative. See the following examples.
\par
\begin{Exa}[The minimal output entropy of an energy-constrained quantum channel]
The \emph{von Neumann entropy} of a quantum state
$\rho$ in $\fS(\cH)$ is a basic characteristic of this state defined by the formula
$H(\rho)=\operatorname{Tr}\eta(\rho)$, where  $\eta(x)=-x\log x$ for $x>0$
and $\eta(0)=0$. The function $H(\rho)$ is  concave and lower semicontinuous on the
set~$\fS(\cH)$ and takes values in~$[0,+\infty]$, see for example \cite{H-SCI,L-2,W}.
\par
A quantum channel from a system $A$ to a system $B$ is a completely positive
trace-preserving linear map $\Phi:\fT(\cH)\to\fT(\cK)$ between the Banach
spaces $\fT(\cH)$ and $\fT(\cK)$, where $\cH$ and $\cK$ are Hilbert spaces
associated with the systems $A$ and $B$, respectively. In the analysis of
information abilities of quantum channels, the notion of the minimal output
entropy of a channel is widely used \cite{H-SCI,GP&Co,MGH,Man,Shor}. It is defined
as
$$
H_{\rm min}(\Phi)=\inf_{\rho\in\fS(\cH)} H(\Phi(\rho))=\inf_{\varphi\in\cH_1} H(\Phi(\ket{\varphi}\!\!\bra{\varphi})),
$$
where $\cH_{1}$ is the unit sphere in $\cH$, and where the second equality follows
from the concavity of the function $\rho\mapsto H(\Phi(\rho))$ and from the
possibility to decompose any mixed state into a convex combination of pure states.
\par
In studies of infinite-dimensional quantum channels, it is reasonable to impose the
energy-constraint on input states of these channels. So, alongside with the minimal
output entropy $H_{\rm min}(\Phi)$, it is reasonable to consider its constrained
versions (cf.~\cite{Man})
\begin{eqnarray}
\label{CME-1}
H_{\rm min}(\Phi,H,E) =
\inf_{\rho\in\fS(\cH):\Tr H\rho\leq E} H(\Phi(\rho)),\\[.5\baselineskip]
\label{CME-2}
H^{=}_{\rm min}(\Phi,H,E) =
\inf_{\rho\in\fS(\cH): \Tr H\rho=E} H(\Phi(\rho)).
\end{eqnarray}
In contrast to the unconstrained case, it is not obvious that the infima in
(\ref{CME-1}) and (\ref{CME-2}) can be taken only over pure states satisfying the
conditions $\Tr H\rho\leq E$ and $\Tr H\rho=E$ correspondingly. In \cite{Man} it
is shown that this holds in the finite-dimensional settings. The above
Proposition~\ref{pro:supremum-pure} allows to prove the same assertion for an
arbitrary infinite-dimensional channel $\Phi$ and any energy observable $H$.
\end{Exa}
\begin{Cor}\label{cor:minimum-output}
Let $H$ be an arbitrary positive operator and let
$E$ be greater than the infimum of the spectrum of $H$. Then both infima in (\ref{CME-1}) and (\ref{CME-2}) can be taken over pure states, i.e.
\begin{eqnarray}
\label{CME+1}
H_{\rm min}(\Phi,H,E) =
\inf_{\varphi\in\cH_1:\,\braket{\varphi|H|\varphi}\leq E} H(\Phi(\ket{\varphi}\!\!\bra{\varphi})),\\[.5\baselineskip]
\label{CME+2}
H^{=}_{\rm min}(\Phi,H,E) =
\inf_{\varphi\in\cH_1:\,\braket{\varphi|H|\varphi}=E} H(\Phi(\ket{\varphi}\!\!\bra{\varphi})).
\end{eqnarray}
If the operator $H$ has unbounded discrete spectrum of finite multiplicity, then the
infimum in (\ref{CME+1}) is attained at a unit vector.
\end{Cor}
\begin{proof}
By Proposition~\ref{pro:supremum-pure} and Remark~\ref{rem:energy-equality}, it
suffices to note that the function $\rho\mapsto H(\Phi(\rho))$ is concave nonnegative
and lower semicontinuous (as a composition of a continuous and a lower 
semicontinuous function).
\end{proof}
Corollary~\ref{cor:minimum-output} simplifies the definitions of the
quantities $H_{\rm min}(\Phi,H,E)$ and $H^{=}_{\rm min}(\Phi,H,E)$ significantly. It
also shows that
$$
H_{\rm min}(\widehat{\Phi},H,E)=H_{\rm min}(\Phi,H,E)\quad \mbox{and}
\quad H^{=}_{\rm min}(\widehat{\Phi},H,E)=H^{=}_{\rm min}(\Phi,H,E),
$$
where $\widehat{\Phi}$ is a complementary channel to the channel $\Phi$, since
for any pure state $\rho$ we have $H(\widehat{\Phi}(\rho))=H(\Phi(\rho))$, see
Section~8.3 of \cite{H-SCI}.
\par
\begin{Exa}[On the definition of the operator E-norms]
On the algebra $\fB(\cH)$ of all bounded operators one can consider the family
$\{\|A\|_E^H\}_{E>0}$ of norms induced by a positive operator $H$ with the infimum
of the spectrum equal to zero \cite{ECN}. For any $E>0$ the norm $\|A\|_E^H$ is
defined as
\begin{equation}\label{ec-on}
 \|A\|^{H}_E\doteq \sup_{\rho\in\fS(\cH):\Tr H\rho\leq E}\sqrt{\Tr A\rho A^*}.
\end{equation}
These norms, called operator \emph{E}-norms, appear as ``doppelganger'' of
the energy-constrained Bures distance between completely positive linear maps in 
the generalized version of the Kretsch\-mann-Schlingemann-Werner theorem 
\cite[Section 4]{ECN}.
\par
For any $A\in\fB(\cH)$ the function $E\mapsto\|A\|_E^H$ is concave and tends to
$\|A\|$ (the operator norm of $A$) as $E\to+\infty$. All the norms $\|A\|_E^H$
are equivalent (for different $E$ and fixed $H$) on $\fB(\cH)$ and generate a
topology depending on the operator $H$. If $H$ is an unbounded operator then
this topology is weaker than the norm topology on $\fB(\cH)$, it coincides with
the strong operator topology on bounded subsets of $\fB(\cH)$ provided that the
operator $H$ has unbounded discrete spectrum of finite multiplicity.
\par
If we assume that the supremum in (\ref{ec-on}) can be taken only over pure states
$\rho$ such that $\Tr H\rho\leq E$ then we obtain the following simpler definition
\begin{equation}\label{ec-on-b}
 \|A\|^H_E\doteq \sup_{\varphi\in\cH_1,\braket{\varphi|H|\varphi}\leq E}\|A\varphi\|,
\end{equation}
which shows the sense of the norm $\|A\|^H_E$ as a constrained version of the operator
norm $\|A\|$. In \cite{ECN} the above assumption was proved only in the case when the
operator $H$ has unbounded discrete spectrum of finite multiplicity.
Proposition~\ref{pro:supremum-pure} (applied to the continuous affine function
$f(\rho)=\Tr A\rho A^*$) allows to fill this gap.
\end{Exa}
\begin{Cor}
For an arbitrary positive operator $H$, the definitions (\ref{ec-on}) and (\ref{ec-on-b})
coincide for any $A\in\fB(\cH)$.
\end{Cor}
\noindent
{\footnotesize
Acknowledgements. 
The first author thanks M.\,R.~Galarza as well as M.\,M. Weis and J.~Weis for hosting 
him while working on this project. The second author is grateful to A.\,S.~Holevo and 
G.\,G.~Amosov for useful discussions.} 
%
%
%
\bibliographystyle{plain}

%
%
%
\vspace{1cm}
\parbox{10cm}{%
Stephan Weis\\
Theisenort 6\\
96231 Bad Staffelstein\\
Germany\\
e-mail \texttt{maths@weis-stephan.de}}
\vspace{\baselineskip}
\par\noindent
\parbox{10cm}{%
Maksim Shirokov\\
Steklov Mathematical Institute\\
Moscow\\
Russia\\
e-mail \texttt{msh@mi.ras.ru}}

\end{document}